\documentclass{article}
\usepackage{array}
\usepackage{amssymb,amsmath}		
\usepackage{graphics}		
\usepackage{longtable}          
\usepackage{subfigure}
\usepackage{float}

\newcommand{\vectornorm}[1]{{\left|\left|#1\right|\right|}_{2}} 
\newcommand{\vectornormOne}[1]{{\left|\left|#1\right|\right|}_{1}} 
\newcommand{\vectornormZero}[1]{{\left|\left|#1\right|\right|}_{0}} 
\newtheorem{theorem}{Theorem}[section]

\newenvironment{proof}[1][Proof]{\begin{trivlist}
\item[\hskip \labelsep {\bfseries #1}]}{$\blacksquare$ \end{trivlist}}

\title{Quantized Network Coding for Sparse Messages}
\author{Mahdy Nabaee and Fabrice Labeau
\thanks{This work was supported by Hydro-Québec, the Natural Sciences and Engineering Research Council of Canada and McGill University in the framework of the NSERC/Hydro-Québec/McGill Industrial Research Chair in Interactive Information Infrastructure for the Power Grid.
}}

\begin{document}

\date{}

\maketitle

\begin{abstract}
In this paper, we study the data gathering problem in the context of power grids by using a network of sensors, where the sensed data have inter-node redundancy.
Specifically, we propose a new transmission method, called \textit{quantized network coding}, which performs linear network coding in the infinite field of real numbers, and quantization to accommodate the finite capacity of edges.
By using the concepts in compressed sensing literature, we propose to use $\ell_1$-minimization to decode the quantized network coded packets, especially when the number of received packets at the decoder is less than the size of sensed data (\textit{i.e.} number of nodes).
We also propose an appropriate design for network coding coefficients, based on restricted isometry property, which results in robust $\ell_1$-min decoding.
Our numerical analysis show that the proposed quantized network coding scheme with $\ell_1$-min decoding can achieve significant improvements, in terms of compression ratio and delivery delay, compared to conventional packet forwarding.
\end{abstract}

\section{Introduction}
Based on the data, reported by North American electric reliability council on 162 disturbances in the power system, problems in real-time monitoring and control system, communication system, and delayed restoration were the major cause of these disturbances \cite{1033736}.
This fact indicates the importance of information infrastructures for reliable and economic operation of power systems.
In the past decade, sensor networks have been widely proposed (and in some cases used) as a promising technology to enhance the future of electric power grid, in different aspects, including power generation, distribution and utilization \cite{Opps}. 

As the primary element of information chain, collection of sensed data by sensor networks is a critical task in monitoring and control of power grids.
Especially, we are interested in data collection in a power substation, where most of the sensed data are naturally (inter-node) correlated.
Motivated by this application, we study {data gathering} scenario (in which messages of different nodes are transmitted to a single node) for correlated sensed data, especially when they are sparse in some transform domain.
Specifically, we compare conventional packet forwarding with network coding and present our proposed network coding based data gathering, followed by a discussion on its theoretical and numerical guarantees.

It is proved that by using packet forwarding via optimal routing, one can achieve the cut set upper bound \cite{netInfFlow} on the information rates of independent sources, in data gathering scenario \cite{NCCorr_NIFwithCOrrSo}.
For the case of correlated sources, optimal distributed source coding and packet forwarding can achieve maximum throughput of the network in data gathering scenario and there is no throughput advantage for using network coding \cite{NCCorr_NIFwithCOrrSo}.
But, this requires the knowledge of source dependencies and network deployment to be available at all nodes, which signifies the importance of using network coding.
Moreover, flexibility and robustness to deployment changes have drawn attention to network coding, as a good alternative for packet forwarding.

Random linear network coding for correlated messages has been studied in \cite{Ho04networkcoding}, for the case of two messages ($n=2$) and an upper bound on error probability of a so called $\alpha$-decoder is derived.
Recently, the idea of using the concepts of \textit{compressed sensing} in data gathering scenario has been proposed in a number of papers \cite{rabbat}\cite{sFeiz}\cite{feizi2011power}, where different applications are considered.
Feizi \textit{et. al.} have proposed the idea of joint source and network coding in which a random mapping is aligned with analogue network coding \cite{katti2007embracing} to decrease temporal and spatial redundancy of sensor data \cite{sFeiz}\cite{feizi2011power}.

Unfortunately, there is not any published result, which discusses the theoretical requirements of local network coding to ensure robust recovery of messages.
We address this by formulating data gathering scenario with network coding and discussing theoretical requirements for robust compressed sensing ($\ell_1$-min) decoding.

Our proposed Quantized Network Coding (QNC) with $\ell_1$-min decoding is formulated in section~\ref{sec:QNC}.
The discussion on the design of local network coding coefficients and theoretical feasibility of $\ell_1$-min decoding, using restricted isometry property is presented in section~\ref{sec:DesignNCodes}, which is followed by
deriving an upper bound on the recovery error of $\ell_1$-min decoder in section~\ref{sec:Decoder}.
Finally, in section~\ref{sec:SimResults}, we present our simulation results and discuss our conclusions in and section~\ref{sec:Conclusions}.

\section{Quantized Network Coding}\label{sec:QNC}
Consider a sensor network, represented by a directed graph, $\mathcal{G}=(\mathcal{V},\mathcal{E})$, where $\mathcal{V}=\{1,\ldots,n\}$ and $\mathcal{E}=\{1,\ldots,|\mathcal{E}|\}$ are the sets of nodes and edges (links), respectively.
Each edge, $e$, can maintain a lossless communication from its tail node, $tail(e)$, to its head node, $head(e)$, without any interface from other edges, at a maximum rate of $C_e$ bits per channel use.
This implies that the input and output contents of edge $e$, at time $t$, represented by $y_e(t)$, are the same. Furthermore, $y_e(t)$ is from a discrete finite alphabet of size $2^{L C_e}$, where $L$ is the block length, used for transmitting over edges.
We define the set of incoming edges to node $v$, as $\textit{In}(v)=\{e: head(e)=v\}$, and
outgoing edges from node $v$, as $\textit{Out}(v)=\{e: tail(e)=v\}$.

We assume that each node includes a random information source, $X_v$, which generates, (random) message $x_v$.
Furthermore, the messages, $\underline{x}=[x_v:v \in \mathcal{V}] \in \mathcal{R}^n$, are such that there is a linear transform matrix, $\phi_{n \times n}$, for which $\underline{x}=\phi \cdot \underline{s}$, and $\underline{s}$ is $k$-sparse (has $k$ non-zero elements).
Referred as \textit{data gathering} scenario, all of the messages, $x_v$'s, are transmitted to a single node, $v_0 \in \mathcal{V}$, called \emph{gateway} (or \textit{decoder}) node.

In networks with finite link capacities, conventional linear network coding should be performed in finite field where the operations are closed in the field \cite{NC_RLNCtoMulticast}.
However, since theoretical results of compressed sensing are developed for the field of real numbers, we propose to perform linear network coding in real field and then quantize the results to couple with the finite capacity of outgoing edges.
Specifically, for all $v \in \mathcal{V}$ and $e \in \textit{Out}(v)$, we define QNC, according to:
\begin{equation}\label{Eq:QNC1}
y_e(t)= \textbf{Q}_e \Big [\sum_{e' \in \textit{In}(v)} \beta_{e,e'}(t)~y_e(t-1)+\alpha_{e,v}(t)~x_v \Big ], 
\end{equation}
where $\textbf{Q}_e[ \centerdot]$ is the quantizer, associated with the outgoing edge $e$, and $\beta_{e,e'}(t)$ and $\alpha_{e,v}(t)$ are the corresponding network coding coefficients, picked from real numbers.
Time index, $t$, is integer and represents the time during which blocks of length $L$, representing quantized network coded packets are transmitted over all edges.
Messages, $x_v$'s, are supposed not to be changing with $t$ and are constant until they are decoded.
Initial rest condition is also assumed to be satisfied in our QNC scenario: $y_e(1)=0,~\forall~e \in \mathcal{E}.$
Representing the quantization error by $n_e(t)$, QNC can be reformulated according to:
\begin{equation}\label{Eq:QNC2}
y_e(t)= \sum_{e' \in \textit{In}(v)} \beta_{e,e'}(t)~y_e(t-1)+\alpha_{e,v}(t)~x_v + n_e(t).
\end{equation}
Moreover, we assume bounded source values, so that: $\forall v \in \mathcal{V}$, $| x_v | < +q_{max}$. 
To have $| y_e(t) | < +q_{max}$, for all $e \in \mathcal{E}$, we pick network coding coefficients such that at each node $v$:
\begin{eqnarray}\label{Eq:antiClipCond1}
\sum_{e' \in \textit{In}(v)} |\beta_{e,e'}(t)|+|\alpha_{e,v}(t)| \leq 1,\forall~v \in \mathcal{V},~\forall~e \in \textit{Out}(v).
\end{eqnarray}

By defining vectors of edge contents, $\underline{y}(t)=[y_e(t):e \in \mathcal{E}]$, and quantization noises, $\underline{n}(t)=[n_e(t):e \in \mathcal{E}]$, we have:
\begin{equation}\label{Eq:matrixForm}
\underline{y}(t)=F(t) \cdot \underline{y}(t-1)+A(t) \cdot \underline{x}+\underline{n}(t),
\end{equation}
where $F(t)$, and $A(t)$ are defined according to:
\begin{equation}
 F(t)_{|\mathcal{E}|\times |\mathcal{E}|}: \{F(t)\}_{e,e'}=\left\{
\begin{array}{l l}
  \beta_{e,e'}(t)  & ,~\scriptsize tail(e)=head(e') \\
  0  &  ,~\mbox{otherwise} \\ \end{array} \right.
\end{equation}
\begin{equation}\label{Eq:defineAt}
 A(t)_{|\mathcal{E}|\times |\mathcal{V}|}: \{A(t)\}_{e,v}=\left\{
\begin{array}{l l}
  \alpha_{e,v}(t)  & ,~tail(e)=v \\
  0  &  ,~\mbox{otherwise} \\ \end{array} \right. .
\end{equation}
By considering (\ref{Eq:matrixForm}) as the difference equation for a linear system with $\underline{n}(t)$'s and $\underline{x}$ as its input, and the content of gateway incoming edges as its output, and using the results in \cite{kailath1980linear}, the \textit{marginal measurements}, $\{z(t)\}_i$, at time $t$, are given by:
\begin{equation}\label{Eq:measForm1}
\underline{z}(t)=[y_e(t):e \in \textit{In}(v_0)]=B \cdot \underline{y}(t)= \Psi(t) \cdot \underline{x}+\underline{n}_{eff}(t),
\end{equation}
where
\begin{eqnarray}
\Psi(t)&=&B \cdot \sum_{t'=2}^{t} \prod_{t''=t}^{t'+1} F(t'')A(t'), \label{Eq:DefPsi} \\
\underline{n}_{eff}(t)&=&B \cdot \sum_{t'=2}^{t} \prod_{t''=t}^{t'+1} F(t'') \cdot \underline{n}(t'), \label{Eq:DefNeff}
\end{eqnarray}
and $B$ is defined such that:
\begin{equation}
\{B\}_{i,e}=\left\{
\begin{array}{l l}
  1  & ,~i~\mbox{corresponds to}~e,~e \in \textit{In}(v_0) \\
  0  &  ,~\mbox{otherwise} \\ \end{array} \right. .
\end{equation}
By storing enough marginal measurements, at the decoder, we build up the \textit{total measurements vector}, $\underline{z}_{tot}(t)$, as follows:
\begin{equation}\label{Eq:totMeasEq}
\underline{z}_{tot}(t)=\left[ {\begin{array}{*{20}c}
	\underline{z}(2) \\	
   \vdots   \\
   \underline{z}(t)   \\
 \end{array} } \right]_{m \times 1}=\Psi_{tot}(t) \cdot \underline{x}+\underline{n}_{eff,tot}(t),
\end{equation}
where the \textit{total measurement matrix}, $\Psi_{tot}(t)$, and \textit{total effective noise} vector, $\underline{n}_{eff,tot}(t)$, are calculated as follows:
\begin{equation}
\Psi_{tot}(t)=\left[ {\begin{array}{*{20}c}
	\Psi(2) \\	
   \vdots   \\
   \Psi(t)   \\
 \end{array} } \right],~~
\underline{n}_{eff,tot}(t)=\left[ {\begin{array}{*{20}c}
	\underline{n}_{eff}(2) \\	
   \vdots   \\
	\underline{n}_{eff}(t)   \\
 \end{array} } \right].\footnote{Since we assume transmission starts from $t=1$, at which initial rest condition holds, $\{z(1)\}_i$'s are all zero and not useful for decoding.}
\end{equation}

It is now desired to recover original messages, $\underline{x}$, from the noisy measurements, $\underline{z}_{tot}(t)$, assuming that enough measurements are stored at the decoder.
Specifically, we are interested in investigating the feasibility of compressed sensing decoding by using $\ell_1$-minimization, when the number of measurements ($\{z_{tot}(t)\}_i, 1 \leq i \leq m$, and $m=(t-1)|\textit{In}(v_0)|$) is less than number of messages, $n$; that is: $m < n$.

\section{Design of Network Coding Coefficients}\label{sec:DesignNCodes}
For compressed sensing with $\ell_1$-min decoding, appropriate measurement matrices should be used \cite{LinProg}.
Specifically, Matrices with good norm conservation properties are shown to be good choices, in this case \cite{candes}.
Restricted Isometry Property (RIP) is defined to characterize this norm conservation. 
Explicitly, $\Theta_{tot}(t)=\Psi_{tot}(t) \phi$ is said to satisfy RIP of order $k$ with constant $\delta_k$ if we have:
\begin{equation}\label{Eq:RIPexact}
1-\delta_k \leq \frac{ \vectornorm{\Theta_{tot}(t)~\underline{s} }^2}{\vectornorm{\underline{s}}^2} \leq 1+ \delta_k,~\forall~\underline{s} \in \mathcal{R}^n,~\vectornormZero{\underline{s}} \leq k.
\end{equation}
It is also shown that random matrices with independently and identically distributed zero mean Gaussian entries of appropriate dimension satisfy RIP with overwhelming probability \cite{simpleProof}.
Such measurement matrices are also global, in the sense that the choice of $\phi$ does not affect the satisfaction of RIP.
In the following, we present a theorem, which summarizes the result of our work on designing local network coding coefficients such that the resulting $\Psi_{tot}(t)$ (and also $\Theta_{tot}(t)=\Psi_{tot}(t) \phi$) is appropriate for compressive sensing.

\begin{theorem}\label{th:designRIP}
If the network coding coefficients, $\alpha_{e,v}(t)$ and $\beta_{e,e'}(t)$, are such that:
\begin{itemize}
\item $\alpha_{e,v}(t)=0,~\forall~t>2,$ and $\alpha_{e,v}(2)$'s are independent zero mean Gaussian random variables,
\item $\beta_{e,e'}(t)$'s are deterministic,
\end{itemize}
then the resulting total measurement matrix, $\Psi_{tot}(t)$, has zero-mean Gaussian entries, and for every $v,v' \in \mathcal{V}$, where $v \neq v'$, $\{\Psi_{tot}(t)\}_{iv}$ and $\{\Psi_{tot}(t)\}_{iv'}$ are independent.
\footnote{The structurally zero entries of $F(t)$'s do not let us have the same variances for entries of $\Psi_{tot}(t)$.} 
\end{theorem}

The proof of this theorem is omitted for lack of space; it is based on manipulation of Gaussian random variable and their linear combinations.

The beauty of this proposed design (based on conditions of theorem~\ref{th:designRIP}) is that there is no need to have an extra overhead communication between the nodes, in order to generate appropriate network coding coefficients.

\section{Recovery Error Bound for $\ell_1$-min Decoder}\label{sec:Decoder}
In the following, we derive an upper bound for the $\ell_2$-norm of recovery error of QNC with $\ell_1$-min decoder (which itself can be implemented by linear programming \cite{LinProg}). 

\begin{theorem}\label{theorem:QNCrecoveryL2norm}
Consider a data gathering scenario with $k$-sparse messages, $\underline{x}$, and sparsity transform $\phi$, where $|x_v| < q_{max},~\forall~v$.
Assume at all nodes, QNC with uniform quantizer of step size $\Delta_{Q,e}$ (possibly different for each $e$) and network coding coefficients $\alpha_{e,v}(t)$ and $\beta_{e,e'}(t)$ is performed. 
Network coding coefficients are also assumed to satisfy the condition of Eq.~\ref{Eq:antiClipCond1}.
In such network, if $\Psi_{tot}(t)$ and $\phi$ are such that $\Psi_{tot}(t) \cdot \phi$ satisfies RIP of order $2k$ with a constant $\delta_{2k} < \sqrt{2}-1$, and $\underline{\hat{x}}(t)$ is the recovery result of $\ell_1$-min decoder of Eq.~\ref{Eq:L1minDecoder} with $\epsilon^2(t)$ calculated according to Eq.~\ref{Eq:epsVal}, then the $\ell_2$-norm of recovery error is upper bounded as in (\ref{Eq:boundError}).
\begin{eqnarray}
\underline{\hat{x}}(t)&=&\phi \cdot \arg\min_{\underline{s}'} \vectornormOne{\underline{s}'}, \label{Eq:L1minDecoder} \\
 && \text{subject to:}~\vectornorm{\underline{z}_{tot}(t)-\Psi_{tot}(t)~ \phi ~ \underline{s}'}^2 \leq \epsilon^2(t) \nonumber
\end{eqnarray}
\begin{eqnarray}
\vectornorm{\underline{x}-\underline{\hat{x}}}^2 & \leq & c_1~\epsilon^2(t) \label{Eq:boundError} \\
c_1 &= & 4\frac{\sqrt{1+\delta_{2k}}}{1-(1+\sqrt{2})\delta_{2k}}  \\
\epsilon^2(t)&=& \frac{1}{4}  \sum_{t'=2}^{t} \Big ( \sum_{t''=1}^{t'-1} \underline{\Delta}_Q^{T}~ {\Big | \prod_{t'''=t''+2}^{t}{F(t''')} \Big |}^{T} \label{Eq:epsVal} \\
&& \cdot B^T B \cdot \sum_{t''=1}^{t'-1} \Big | \prod_{t'''=t''+2}^{t'}{F(t''')} \Big | ~ \underline{\Delta}_Q \Big )\nonumber \\
\underline{\Delta}_Q &=& [\Delta_{Q,e}:e \in \mathcal{E}] 
\end{eqnarray}
\end{theorem}
\begin{proof}
Since the network is lossless and network coding coefficients satisfy the condition of Eq.~\ref{Eq:antiClipCond1}, and $|x_v| \leq {q}_{max},~\forall~v$, the only associated noise is quantization noise at each edge.
As a result of having uniform quantizers, we have:
$| n_{e}(t) | \leq \frac{\Delta_{Q,e}}{2},~\forall~e \in \mathcal{E}.$
Equivalently, the absolute value vector of $\underline{n}(t)$, represented by $|\underline{n}(t)|$, is such that:
$| \underline{n}(t)| \leq \frac{1}{2} \underline{\Delta}_Q.$
Therefore, $|\underline{n}_{eff}(t)|$ can be upper bounded as follows:
\begin{eqnarray}
\Big | \underline{n}_{eff}(t) \Big| & \leq & B \cdot  \sum_{t'=1}^{t-1} \Big | \prod_{t''=t'+2}^{t}{F(t'')} \Big | \cdot \frac{1}{2} \underline{\Delta}_Q.  \label{Eq:absUpB5} 
\end{eqnarray}
This implies:
\begin{eqnarray}
\vectornorm{\underline{n}_{eff,tot}(t)}^2 &=&  \sum_{t'=2}^{t} \vectornorm{\underline{n}_{eff}(t')}^2 \label{Eq:uppBoundNoise} \\
&=& \sum_{t'=2}^{t} {\Big | \underline{n}_{eff}(t') \Big |}^{T}  \Big | \underline{n}_{eff}(t') \Big |
 \leq  \epsilon^{2}(t), \nonumber
\end{eqnarray}
where $\epsilon^2(t)$ is as in Eq.~\ref{Eq:epsVal}.
Now, by applying theorem~1.2 in \cite{candes} and using (\ref{Eq:uppBoundNoise}), theorem is proved.
\end{proof}

Since for uniform quantizers $\Delta_{Q,e}$ is equal to ${2{q}_{max}}/({2^{L C_e}-1})$, the upper bound in the theorem, $c_1 \epsilon$ is decreases when the block length, $L$, is increased.
However, this will result in an undesirable increase on the delivery delay in the network, which requires us to find the optimal $L$ for each quality of service (\textit{i.e.} SNR).

\section{Simulation Results}\label{sec:SimResults}
We evaluate the performance of QNC in terms of average Signal to Noise Ratio (SNR), $\overline{SNR}(t)=10\log_{10} ({\overline{\vectornorm{\underline{x}}}}/{\overline{\vectornorm{\underline{x}-\underline{\hat{x}}(t)}}})$, and average delivery delay, $\overline{\tau}_d(t)=L \overline{(t-1)}$, where the averaging is done over different realizations of simulation runs, \textit{i.e.} network deployments.

To set up the simulations, we randomly generate networks of nodes, in which directed edges with unit capacity, \textit{i.e.} $C_e=1~[\frac{bit}{channel~ use}],~\forall~e$, are randomly spread between different pairs of nodes. 
One of the nodes is randomly picked to be the gateway node, $v_0$, at which the messages are recovered.
Simulations are repeated by generating $150$ different random realizations of network deployments to obtain smoothed results.

For each generated random network deployment, we perform QNC for different values of message sparsity factor; \textit{i.e.} $\frac{k}{n}=0.1,0.2,0.3$.
Specifically, to generate messages, $\underline{x}$, we first generate a $k$-sparse random vector, $\underline{s}$, whose components are uniformly distributed between $-\frac{1}{2}$ and $+\frac{1}{2}$.
This is followed by generation of an orthonormal random matrix, $\phi$, calculating $\phi \cdot \underline{s}$ and normalizing the results between $-q_{max}$ and $+q_{max}$ to obtain $\underline{x}$.
Moreover, network coding coefficients, $\alpha_{e,v}(t)$ and $\beta_{e,e'}(t)$, are generated according to the conditions of theorem~\ref{th:designRIP} and normalizing the results to satisfy the condition of Eq.~\ref{Eq:antiClipCond1} and prevent overflow.
At the decoder, the received measurements up to $t$, $\underline{z}_{tot}(t)$, are used to recover $\underline{\hat{x}}(t)$, according to (\ref{Eq:L1minDecoder}), which is calculated by using the open source implementation in \cite{cvx1}.

For each deployment, we also simulate a routing based packet forwarding and compare its performance with QNC.
To find the routes from each node to the gateway node, we calculate the shortest path from each node to the gateway node, using the Dijkstra algorithm \cite{dijkstra1959note}.
Obviously, in this case, the only associated error for received packets (messages) at the decoder is the quantization noise at the source nodes.

\begin{figure*}[t]
\centering
\subfigure[$1400$ edges]{
\resizebox{.69\textwidth}{!}{
\includegraphics{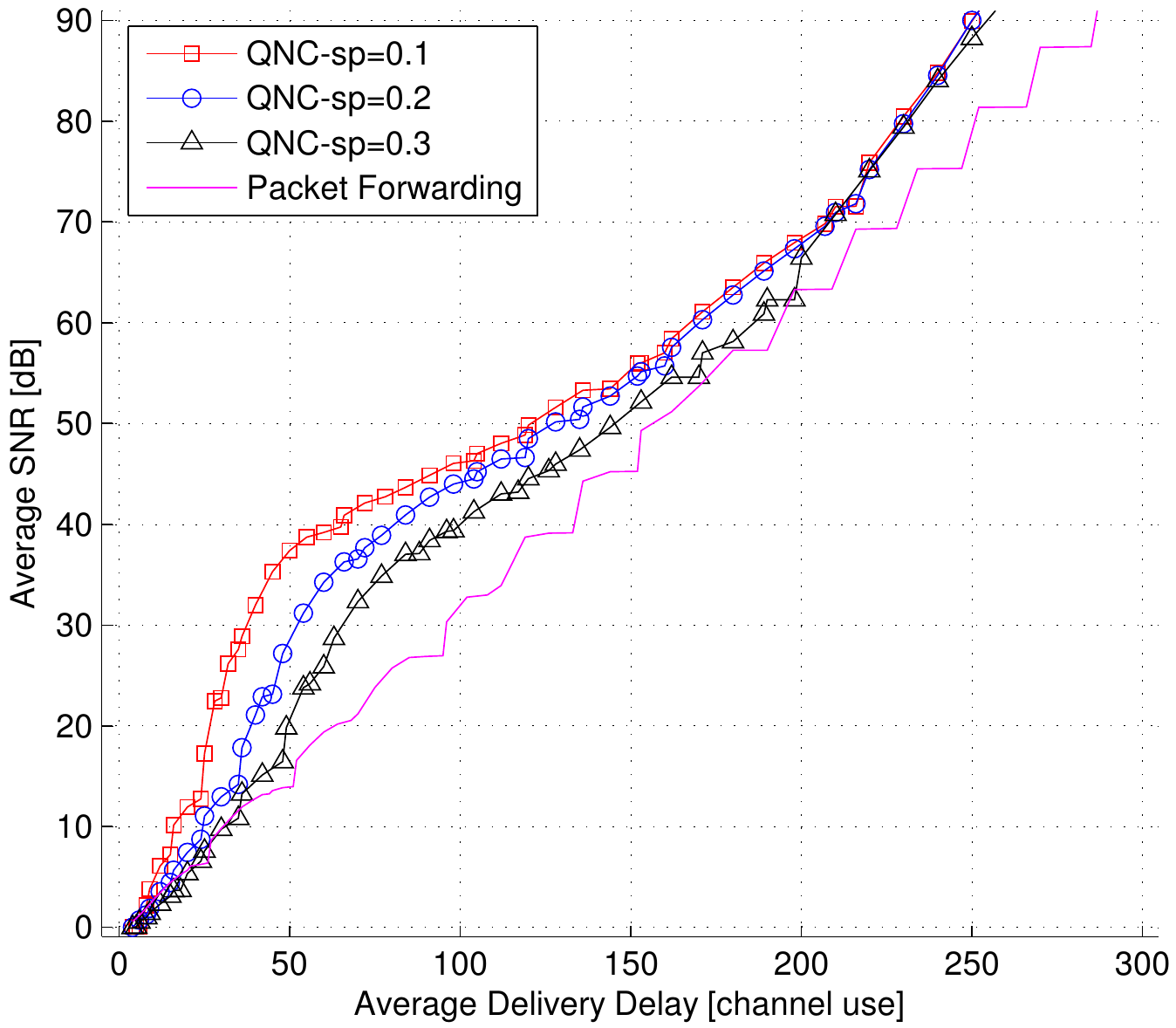}}
\label{fig:subfig1}
} \qquad
\subfigure[$1100$ edges]{
\resizebox{.69\textwidth}{!}{
\includegraphics{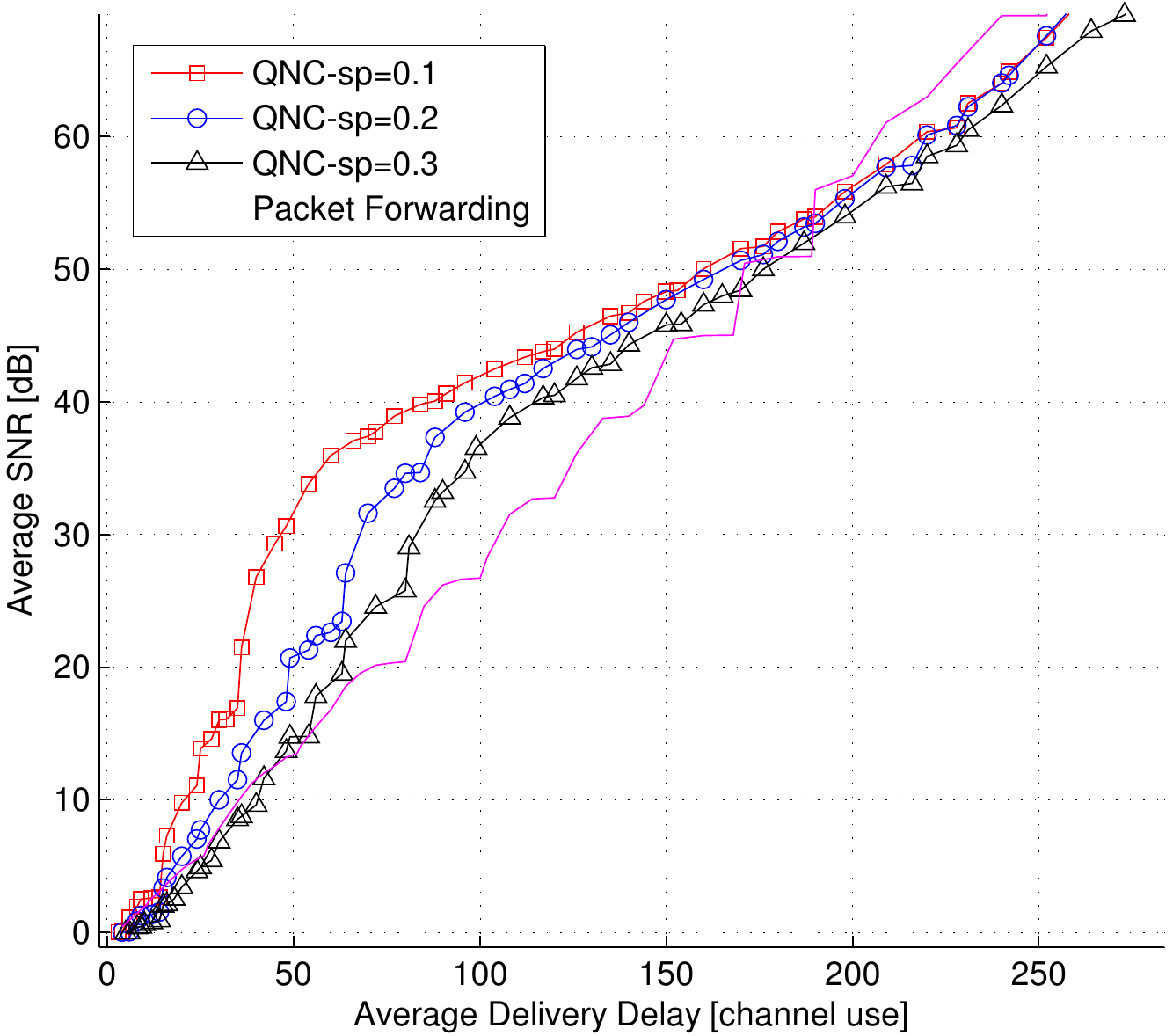}}
\label{fig:subfig2}
} \\
\caption{Average SNR versus average delivery delay of QNC and Packet Forwarding for \subref{fig:subfig1} $1400$ \subref{fig:subfig2} $1100$ edges and $\frac{k}{n}=0.1,0.2,0.3$\label{fig:subfigureExample}.
}
\end{figure*}

In each case of QNC and packet forwarding, we find the optimal block length, $L$, by repeating the simulations for different integer values of $L$ and then picking the smallest $\overline{\tau}_d(t)$ for each value of $\overline{SNR}(t)$.
The average SNR - average delivery delay curves, resulting from optimizing the block length for QNC and packet forwarding, are depicted in Fig.~\ref{fig:subfigureExample}.
As shown, in all of the cases, the proposed QNC with $\ell_1$-min decoding outperforms packet forwarding in small delivery delays, which corresponds to small number of received measurements at the decoder.
However, when the number of received measurement passes a certain threshold and increases, the SNR of $\ell_1$-min decoder can not beat that of packet forwarding. Even for $1100$ edges in Fig.~\ref{fig:subfig2}, QNC fails to achieve the same SNR as packet forwarding.
To explain this, we should note that for higher SNR values, we may need to collect more measurements at the decoder, which requires a larger $t$.
Since a larger $t$ is picked for decoding the messages, more quantization noises are contributed to the measurements, which increases the $\ell_2$-norm of effective measurement noise.

\section{Conclusions and Future Works}\label{sec:Conclusions}
Joint source and network coding of sparse messages was discussed in this paper.
As a good alternative for packet forwarding, linear network coding in the real field, is coupled with quantization, to gather sparse data in a decoder node.
Moreover, the required conditions for theoretical guarantee of $\ell_1$-min recovery was discussed and an appropriate design for network coding coefficients, in terms of RIP, was proposed.
Finally, in section~\ref{sec:SimResults}, by using simulations, we have shown the promising performance of $\ell_1$-min recovery for quantized network coded packets, in terms of SNR versus delivery delay.
As our future plan, we are interested to derive theoretical guarantees for satisfaction of RIP of $\Psi_{tot}(t)$, resulted from our designed network coding coefficients.
It is also planned to apply our QNC with $\ell_1$-min decoding to realistic data gathering scenarios in power substations.


\bibliographystyle{IEEEbib}
\bibliography{Ref}

\begin{thebibliography}{10}

\bibitem{1033736}
Zhaoxia Xie, G.~Manimaran, V.~Vittal, A.G. Phadke, and V.~Centeno,
\newblock ``An information architecture for future power systems and its
  reliability analysis,''
\newblock {\em Power Systems, IEEE Transactions on}, vol. 17, no. 3, pp. 857 --
  863, aug. 2002.

\bibitem{Opps}
V.C. Gungor, Bin Lu, and G.P. Hancke,
\newblock ``Opportunities and challenges of wireless sensor networks in smart
  grid,''
\newblock {\em IEEE Transactions on Industrial Electronics}, vol. 57, no. 10,
  pp. 3557 --3564, 2010.

\bibitem{netInfFlow}
R.~Ahlswede, Ning Cai, S.-Y.R. Li, and R.W. Yeung,
\newblock ``Network information flow,''
\newblock {\em IEEE Transactions on Information Theory}, vol. 46, no. 4, pp.
  1204 --1216, July 2000.

\bibitem{NCCorr_NIFwithCOrrSo}
J.~Barros and S.D. Servetto,
\newblock ``Network information flow with correlated sources,''
\newblock {\em IEEE Transactions on Information Theory}, vol. 52, no. 1, pp.
  155 -- 170, 2006.

\bibitem{Ho04networkcoding}
T.~Ho, M.~M{\'e}dard, M.~Effros, R.~Koetter, and DR~Karger,
\newblock ``Network coding for correlated sources,''
\newblock in {\em Proceedings of Conference on Information Sciences and
  Systems}, 2004.

\bibitem{rabbat}
J.~Haupt, W.U. Bajwa, M.~Rabbat, and R.~Nowak,
\newblock ``Compressed sensing for networked data,''
\newblock {\em IEEE Signal Processing Magazine}, vol. 25, no. 2, pp. 92 --101,
  march 2008.

\bibitem{sFeiz}
S.~Feizi, M.~M{\'e}dard, and M.~Effros,
\newblock ``Compressive sensing over networks,''
\newblock in {\em Communication, Control, and Computing (Allerton), 2010 48th
  Annual Allerton Conference on}. IEEE, 2010, pp. 1129--1136.

\bibitem{feizi2011power}
S.~Feizi and M.~Medard,
\newblock ``A power efficient sensing/communication scheme: Joint
  source-channel-network coding by using compressive sensing,''
\newblock {\em Arxiv preprint arXiv:1110.0428}, 2011.

\bibitem{katti2007embracing}
S.~Katti, S.~Gollakota, and D.~Katabi,
\newblock ``Embracing wireless interference: Analog network coding,''
\newblock in {\em ACM SIGCOMM Computer Communication Review}. ACM, 2007,
  vol.~37, pp. 397--408.

\bibitem{NC_RLNCtoMulticast}
T.~Ho, M.~Medard, R.~Koetter, D.R. Karger, M.~Effros, Jun Shi, and B.~Leong,
\newblock ``A random linear network coding approach to multicast,''
\newblock {\em IEEE Transactions on Information Theory}, vol. 52, no. 10, pp.
  4413 --4430, 2006.

\bibitem{kailath1980linear}
T.~Kailath,
\newblock {\em Linear systems}, vol.~1,
\newblock Prentice-Hall Englewood Cliffs, NJ, 1980.

\bibitem{LinProg}
E.J. Candes and T.~Tao,
\newblock ``Decoding by linear programming,''
\newblock {\em IEEE Transactions on Information Theory}, vol. 51, no. 12, pp.
  4203 -- 4215, December 2005.

\bibitem{candes}
Emmanuel~J. Candès,
\newblock ``The restricted isometry property and its implications for
  compressed sensing,''
\newblock {\em Comptes Rendus Mathematique}, vol. 346, no. 9-10, pp. 589 --
  592, 2008.

\bibitem{simpleProof}
Richard Baraniuk, Mark Davenport, Ronald Devore, and Michael Wakin,
\newblock ``A simple proof of the restricted isometry property for random
  matrices,''
\newblock {\em Constr. Approx}, vol. 2008, 2007.

\bibitem{cvx1}
M.~Grant and S.~Boyd,
\newblock ``{CVX}: Matlab software for disciplined convex programming, version
  1.21,'' {http://cvxr.com/cvx}, 2011.

\bibitem{dijkstra1959note}
E.W. Dijkstra,
\newblock ``A note on two problems in connexion with graphs,''
\newblock {\em Numerische mathematik}, vol. 1, no. 1, pp. 269--271, 1959.

\end{thebibliography}
\end{document}